\documentclass[a4paper,11pt]{article}

\usepackage[left=2.5cm,right=2.5cm,top=3cm,bottom=3cm]{geometry}

\usepackage{amsmath,amssymb,amsthm,enumerate}
\usepackage{enumitem}
\usepackage{mathtools}
\usepackage{hyperref}
\usepackage{xcolor}
\usepackage{bm}
\usepackage{standalone}
\usepackage{comment}
\usepackage{listings} 
\usepackage{orcidlink}
\usepackage{upgreek}
\usepackage{comment}
\usepackage{mathrsfs}

\newtheorem{theorem}{Theorem}
\numberwithin{theorem}{section}
\newtheorem{corollary}[theorem]{Corollary}

\theoremstyle{definition}
\newtheorem{definition}[theorem]{Definition}
\newtheorem{remark}[theorem]{Remark}

\newcommand{\beq}{\begin{equation}}
\newcommand{\abs}[1]{|#1|}
\newcommand{\eeq}{\end{equation}}
\newcommand{\SO}{\mathrm{SO}}
\newcommand{\midd}{\textup{s.t.}}
\newcommand{\diag}{\operatorname{diag}}
\newcommand{\cR}{\mathcal{R}}
\newcommand{\QQ}{\mathbb{Q}}
\newcommand{\ZZ}{\mathbb{Z}}
\newcommand{\UU}{\mathbb{U}}
\newcommand{\HH}{\mathbb{H}}
\newcommand{\RR}{\mathbb{R}}
\newcommand{\mM}{\mathsf{M}}
\newcommand{\ii}{\mathbf{i}}
\newcommand{\jj}{\mathbf{j}}
\newcommand{\kk}{\mathbf{k}}
\newcommand{\N}{\mathbb{N}}

\newcommand{\rn}{\mathrm{nrd}}
\newcommand{\ch}{\mathrm{char}}
\newcommand{\Id}{\mathrm{Id}}
\newcommand{\Hi}{\mathbf{H}}
\newcommand{\Nrd}{\mathrm{nrd}}
\newcommand{\Trd}{\mathrm{trd}}
\newcommand{\F}{\mathbb{F}}
\newcommand{\Al}{\mathfrak{A}}
\newcommand{\End}{\mathrm{End}_{\small\mbox{$\QQ_{p,v}$}}}
\DeclareMathAlphabet{\mathdutchcal}{U}{dutchcal}{m}{n}

\allowdisplaybreaks 

\definecolor{darkgreen}{rgb}{0.0, 0.8, 0.5}   %for Ilaria's notes
\begin{document}

\title{The Haar measure of the $p$-adic rotation group $\SO(3)_p$\\ via nautical angles}

\author{
Lorenzo Guglielmi,$^{1,2}$\thanks{Email: lorenzo.guglielmi@unicam.it}\,\orcidlink{0009-0006-1624-049X} \hspace{1.2mm}
Stefano Mancini,$^{1,2}$\thanks{Email: stefano.mancini@unicam.it}\, \orcidlink{0000-0002-3797-3987} \hspace{1.2mm}
Vincenzo Parisi,$^{3}$\thanks{Email: vincenzo.parisi@iit.it}\, \orcidlink{0000-0001-9563-6471} \hspace{1.2mm} and Ilaria Svampa,$^{4}$\thanks{Email: ilaria.svampa@uni-koeln.de}\, \orcidlink{0000-0002-1389-0319}}

\date{
{\normalsize
$^1$\textit{School of Science and Technology, University of Camerino,\\
via Madonna delle Carceri, 9, Camerino, I-62032, Italy}\\
\vspace{3mm}
$^2$ \textit{Istituto Nazionale di Fisica Nucleare, Sezione di Perugia,\\
via A.~Pascoli, I-06123 Perugia, Italy}\\
\vspace{3mm}
$^3$ \textit{CONCEPT Lab, Fondazione Istituto Italiano di~Tecnologia,\\
via E.~Melen 83, Genova, 16152, Italy}\\
\vspace{3mm}
$^4$ \textit{Department Mathematik/Informatik--Abteilung Informatik, Universit\"at zu K\"oln, Albertus-Magnus-Platz, 50923 K\"oln, Germany}
}}

\maketitle
	
\begin{abstract}
\noindent We study the explicit construction of the Haar measure on the compact $p$-adic rotation group $\SO(3)_p$ by nautical (Cardano) parametrization. Exploiting its topological group isomorphism with $\HH_p^\times/\QQ_p^\times$ of $p$-adic quaternions modulo scalars, we derive the corresponding change of variables formulas and compute the associated Jacobian in the $p$-adic setting, which we combine with the known Haar measure on the multiplicative group of $p$-adic quaternions $\mathbb{H}_p^\times$. This yields an explicit formula for the normalized Haar measure on $\SO(3)_p$ in nautical coordinates, with a factorized density in the three angles. Our construction provides a concrete tool suited for applications of non-Archimedean models where an explicit angular description of invariant integration is required.

\end{abstract}

\tableofcontents

%%====================================================================================

\section{Introduction}

Rotation groups play a central role in applied mathematics, and it is no coincidence that the special orthogonal groups, denoted $\SO(n)_{\RR}$, rank among the most thoroughly explored and widely recognized mathematical constructs. The case $n=3$ stands out for its rich algebraic structure and the ease with which its actions on Euclidean space can be visualized. Actually, $\SO(3)_{\RR}$ encompasses all rotations around axes in three-dimensional space, with each element admitting an essentially unique decomposition into Euler or nautical angles\footnote{Nautical angles are also known as Cardano angles or as Tait-Bryan angles.}.

The prominence of $\SO(n)_{\RR}$ is closely tied to the natural assumption that physical space corresponds to the Euclidean space $\RR^3$, which itself is built upon the real numbers $\RR$ as the metric completion of the rationals $\QQ$. Yet, it's important to recognize that $\RR$ is not the only possible completion of $\QQ$. According to Ostrowski’s theorem, there exists a countable family of alternative completions—specifically, the $p$-adic number field $\mathbb{Q}_p$, where $p$ is a prime \cite{serre1973}. The fractal-like geometry of $p$-adic numbers has attracted interest in theoretical physics, particularly in proposals for “non-Euclidean” space-time models at extremely small (or large) scales, as might be relevant in quantum gravity or string theory \cite{Araf,FW,Vol}. More recently, there has been also a growing attention to the potential applications of $p$-adic fields in quantum information theory \cite{our2nd,AMP23,SIMW26,aniello2024quantum,vp}.

Given the central role of $p$-adic numbers in number theory, it is unsurprising that rotation groups have been examined over the fields $\QQ_p$. Special orthogonal groups in the $p$-adic setting are constructed from quadratic forms over $\QQ_p$. Unlike the real case, however, definite quadratic forms over $\QQ_p$ exist only in dimensions two, three, and four \cite{our1st}. As a result, the groups $\SO(2)_p$,  $\SO(3)_p$, and $\SO(4)_p$ are the only compact special orthogonal groups in the $p$-adic context. Notably, $\SO(3)_p$ can be viewed as the group of rotations in $\QQ_p^3$, and its geometric structure has been investigated in \cite{our1st}.

Because these groups are compact, the Peter–Weyl theorem guarantees that every irreducible unitary representation is realized as a sub-representation within the regular representation \cite{folland2016course}. The analysis of the regular representation of compact groups—and other core elements of abstract harmonic analysis—fundamentally depends on the Haar measure, a uniquely determined Radon measure on the group. Interpreted as a functional, this measure is known as the Haar integral \cite{folland99}.
In \cite{our3rd}, the Haar integral on $\SO(3)_p$ was investigated using a Lie group–theoretical approach, based on the construction of a suitable atlas of local charts. In \cite{inverse}, the Haar measure on $\SO(3)_p$ was derived through a different, more universal framework—namely, the inverse limit of measure spaces. However, for practical applications, it is preferable to express this measure in terms of angles that characterize rotations, analogous to the familiar case in $\RR^3$. This is the objective we pursue here. Since Euler angles are not available in the $p$-adic setting, we instead employ Cardano decompositions and nautical angles existing for any $p>2$ \cite{our1st}. We then establish a topological and group-theoretic correspondence between $\SO(3)_p$ and the multiplicative group of $p$-adic quaternions, which ultimately enables us to express the Haar measure on $\SO(3)_p$ explicitly in terms of nautical angles.

The structure of the paper is as follows. Section 2 provides a review of $p$-adic rotation groups and introduces the Cardano decomposition in terms of nautical angles for $\SO(3)_p$. In Section 3, we revisit the $p$-adic quaternion algebra, emphasizing the topological and group-theoretic relationship between $\SO(3)_p$ and the multiplicative group of $p$-adic quaternions. Section 4 marks the beginning of our original contributions, where we construct the explicit isomorphism between $p$-adic quaternions and $p$-adic rotations. This construction enables us, in Section 5, to derive the Haar measure on $\SO(3)_p$ in terms of nautical angles. Finally, Section 6 presents our conclusions.

\section{$p$-Adic rotation groups}
In this section, we establish our main conventions and recall some preliminary facts concerning the $p$-adic special orthogonal groups, that will be constantly used throughout this work. At the most general level, a special orthogonal group arises from a distinguished class of linear symmetries of a quadratic form $Q\colon V\rightarrow \F$ defined over an $\F$-vector space $V$. Indeed, if $\F$ has characteristic $\ch(\F)\neq 2$ (as it is for $\mathbb{R}$ and $\QQ_p$), there is a one-to-one correspondence %relation 
between (non-degenerate) quadratic forms and bilinear forms on $V$ \cite{quat2}. The special orthogonal group $\SO(V)$ is then defined as the group of all the \emph{special} ---i.e., with determinant equal to 1--- \emph{isometries} of $Q$. Equivalent quadratic forms lead to isomorphic special orthogonal groups. A special orthogonal group $\SO(V)$ is \emph{compact} if and only if its associated quadratic form is non-isotropic, which corresponds to definite over $\mathbb{R}$ (we will use both terms interchangeably over $\mathbb{Q}_p$). In this work, we shall be mainly interested in the case where $V\equiv \QQ_p^n$, and we will consider $n=2,3,4$ only, as definite quadratic forms over $\QQ_p$ exist solely in these dimensions \cite{our1st}. %We recall that 

An element $\eta$ in $\QQ_p$ is called a $p$-adic unit if $\eta\in\QQ_p^\times$, and $\abs{\eta}_p=1$; adhering to the standard notation, we shall denote by $\UU_p$ the group of $p$-adic units in $\QQ_p$ ---that is, the group of invertible elements in the ring $\ZZ_p$ of $p$-adic integers. A non-quadratic $p$-adic unit is any $u\in\UU_p$ such that $u\notin(\QQ^\times_
p)^2$. The quotient group $\QQ^\times_p/(\QQ^\times_p)^2$ consists of four ‘square classes’, for $p\neq 2$, with
representatives $\{1, u, p, up\}$ ---where $u$ is any non-quadratic unit of $\QQ_p$--- whereas,
for $p = 2$, it consists of eight square classes, with representatives $\{1, 2, 3, 5, 6, 7, 10, 14\}$ (or,
equivalently, $\{\pm 1,\pm 2,\pm 3,\pm 6\}$). 
Let $p>2$ be a prime number. Using the square classes in $\QQ^\times_p/(\QQ^\times_p)^2$, one can give a complete classification of the definite quadratic forms over $\QQ_p$~\cite{serre1973}. Indeed, let $u\in\UU_p$ be a non-quadratic unit, and let $v\in\UU_p$ be defined as
\begin{equation}\label{v}
    \UU_p\ni v\coloneqq
\begin{cases}
 -1 &\text{if}\;\; p\equiv3 \mod4,\\
 -u &\text{if}\;\; p\equiv1 \mod4.
\end{cases}
\end{equation}
Then, up to linear equivalence and scaling, one can show there are precisely three non-isotropic quadratic forms over $\QQ_p^2$ for every prime $p>2$, i.e., 
\begin{equation}\label{eq:quadrform2}
Q_{-v}(\bm{x})= x_1^2-vx_2^2, \quad Q_p(\bm{x})=x_1^2+px_2^2,\quad Q_{up}(\bm{x})= u x_1^2+px_2^2,
\end{equation}
as we all as a unique definite quadratic form $Q_+(\bm{x})$ on $\QQ_p^3$, and $Q_+^{(4)}$ on $\QQ_p^4$:
\begin{equation}\label{eq:quadrform3}
Q_+(\bm{x})=x_1^2-vx_2^2+px_3^2,\quad Q_+^{(4)}(\bm{x})=x_1^2 - v x_2^2 + p x_3^2 - vp x_4^2.
\end{equation}
The special orthogonal groups in dimension two, three and four can now be defined as the groups of the special linear symmetries of the quadratic forms in Eq.\eqref{eq:quadrform2} and Eq.\eqref{eq:quadrform3}; that is, we have the following
\begin{definition}[\cite{our1st}]
%Denoting by 
Let $A_d \coloneqq \diag(1,d)$, $d\in\{-v, p, up\}$, $A_+\coloneqq \diag(1,-v,p)$, and $A_+^{(4)}\coloneqq\diag(1,-v,p,-vp)$ be the matrix representation (w.r.t.\ the canonical basis of $\QQ_p^n$, $n=2,3,4$) of the non-isotropic quadratic forms in $\QQ_p^2$, $\QQ_p^3$ and $\QQ_p^4$, respectively for every prime $p>2$.
Then, we can define, up to isomorphism, the following compact $p$-adic special orthogonal groups:
\begin{align} 
&\SO(2)_{p,d}\coloneqq\left\{L\in \mM(2,\QQ_p)\mid  L^\top A_d L=A_d,\ \det(L)=1\right\},\nonumber\\
& \SO(3)_p\coloneqq \left\{L\in \mM(3,\QQ_p)\mid L^\top A_+L=A_+,\ \det(L)=1\right\},\nonumber\\
&\SO(4)_p\coloneqq \left\{L\in \mM(4,\QQ_p)\mid  L^\top A_+^{(4)}L=A_+,\ \det(L)=1\right\},
\end{align}
where we denoted by $\mM(n,\QQ_p)$ the associative algebra of $n\times n$ matrices over the field of $p$-adic numbers $\QQ_p$.
\end{definition}

From Theorem~$12$ in~\cite{our1st}, we have a complete characterization of the special orthogonal group $\SO(2)_{p,d}$. Indeed, any element of $\SO(2)_{p,d}$ admits the following one-to-one parametrization in terms of the $p$-adic projective line:
\begin{equation}\label{eq:rotazgeneric2} 
\QQ_p\cup\{\infty\}\ni\alpha\mapsto \cR_d(\alpha)=
\begin{pmatrix}
  \frac{1-d\alpha^2}{1+d\alpha^2}
    & -\frac{2d\alpha}{1+d\alpha^2}\\
  \frac{2\alpha}{1+d\alpha^2}
    & -\frac{1-d\alpha^2}{1+d\alpha^2}
\end{pmatrix}\in \SO(2)_{p,d},
\end{equation}
where $\cR_d(\infty)=-\cR_d(0)=-I_2$ ---i.e., the identity element in $\mM(2,\QQ_p)$. The composition of any two rotations in $\SO(2)_{p,d}$ is provided, for fixed $d$, by 
\begin{equation}\label{eq:complowSO2p}
\cR_d(\alpha)\cR_d(\beta)=\cR_d\bigg(\frac{\alpha+\beta}{1-d\alpha\beta}\bigg),
\end{equation}
for every $\alpha,\beta\in\QQ_p\cup\{\infty\}$.

\begin{remark}
We notice that for every prime $p>2$ and $d\in\{-v,p,-vp\}$ there exists $\bm{n}\in\QQ_p^3\setminus\{\mathbf{0}\}$ such that $\SO(3)_{p,\QQ_p\bm{n}}\simeq \SO(2)_{p,d}$, where we denoted by $\SO(3)_{p,\bm{n}}$ the subgroup of $\SO(3)_p$ of rotations around $\QQ_p\mathbf{n}$ --- see Proposition~$21$ in~\cite{our1st}. Specifically, for the rotations around the reference axes $x=\QQ_p\mathbf{e}_1,\,y=\QQ_p\mathbf{e}_2$ and $z=\QQ_p\mathbf{e}_3$ of $\QQ_p^3$, we have:
\begin{equation}
\SO(3)_{p,x}\simeq\SO(2)_{p,-p/v},\quad \SO(3)_{p,y}\simeq \SO(2)_{p,p},\qquad \SO(3)_{p,z}\simeq \SO(2)_{p,-v}.
\end{equation}
\end{remark}
The groups $\SO(2)_{p,d}$, $\SO(3)_p$ and $\SO(4)_p$ are topological groups once endowed with the natural ultrametric topology induced by the ultrametric norm $\|\cR\|_p\equiv\|(\cR_{ij})_{i,j}\|\coloneqq \max_{i,j}\abs{\cR_{ij}}_p$. For every prime $p>2$, the entries of the matrices of these groups are in $\mathbb{Z}_p$ and then these groups are bounded with respect to the $p$-adic norm. They are also closed and hence compact --- see Corollary II$.12$ in \cite{inverse} and Theorem $2.19$ in \cite{thesis}.
Here we stress that our assumption of considering only the special orthogonal groups associated with non-isotropic quadratic forms cannot be dispensed with, since those groups defined on indefinite quadratic forms are not bounded, whence, not compact.

\begin{remark}
The $p$-adic special orthogonal groups in dimension two, three and four are special instances of second countable \emph{$p$-adic Lie groups}. As such, they are totally disconnected, locally compact, Hausdorff spaces. Moreover, as topological spaces, they are \emph{strictly paracompact} ---i.e., every open cover of $\SO(2)_{p,d}$, and $\SO(n)_p$, $n=3,4$, admits a refinement consisting of pairwise disjoint open sets--- as well as \emph{Polish}, since they admit a second countable topology~\cite{our3rd}. 
\end{remark}

From Theorem~22 and Corollary~23 in \cite{our1st} we have that for every prime $p>2$, any $\cR\in \SO(3)_p$ can be written as any of the compositions 
\beq\label{eq:Cardanodecos}
\cR_z\cR_y\cR_x, \quad \cR_z\cR_x\cR_y, \quad \cR_x\cR_y\cR_z, \quad \cR_y\cR_x\cR_z,
\eeq
respectively by certain angles $\alpha,\beta,\gamma\in\QQ_p\cup\infty$. These angles are called \emph{Tait-Bryan}
angles or \emph{nautical} angles or \emph{Cardano} angles.
Moreover, from Remark~27 in \cite{our1st}, no Euler decomposition exists for $\SO(3)_p$ for primes $p>2$. From Remark~28 in \cite{our1st}, for $p=2$, none of the nautical and Euler decompositions exist for $\SO(3)_2$.

\section{$p$-Adic quaternions}
Our next aim is to recall some fundamental facts concerning the quaternion algebra $\HH_p$ over the field $\QQ_p$ of $p$-adic numbers~\cite{quat2,kochubei2001pseudo,quat1}.

We start with the  following 

\begin{definition}\label{def.3.1}

Let $p>2$ be a prime number. By the\textit{ $p$-adic quaternion algebra} over $\QQ_p$ we mean the (division) algebra $\HH_p$ with $\QQ_p$-basis $\mathfrak{B}=(\Id, \ii, \jj, \kk \coloneqq \ii \jj )$ ---where $\Id$ is the identity element in $\HH_p$--- satisfying the following multiplication rules, 
\beq \label{eq.quatalgodd}
\ii^2=v, \ \jj^2=-p, \ \jj\ii=-\ii\jj,
\eeq
where $v\in\UU_p$ is the $p$-adic quadratic unit defined in Eq.\eqref{v}.
\end{definition}
Definition~\ref{def.3.1} of $p$-adic quaternion algebras is actually \emph{symmetric} in $v$ and $-p$, in the sense that if $\HH_p^\prime$ is the quaternion algebra with $\QQ_p$-basis $\mathfrak{B}$, such that $\ii^2 = -p$ and $\jj^2 = v$, then $\HH_p^\prime \simeq \HH_p$, i.e., the map interchanging $\ii$ and $\jj$ is an algebra isomorphism. 
\begin{remark}
A quaternion algebra over $\QQ_p$ can be equivalently defined by saying that the division $\QQ_p$-algebra $\HH_p$ is a quaternion algebra if there exist $\ii,\jj$ in $\HH_p$ which generate $\HH_p$ as a $\QQ_p$-algebra, and satisfy conditions~\eqref{eq.quatalgodd}. From this, it is then automatically true that $\HH_p$ has dimension $4$ as a $\QQ_p$-vector space, and admits the $\QQ_p$-basis $\mathfrak{B}=\{\Id, \ii, \jj, \kk\coloneqq \ii \jj\}$. 
\end{remark}
From the very definition, it is clear that $\HH_p\simeq \QQ_p \times \QQ_p^3$. In particular, any quaternion can be expressed as 
\begin{equation}
    \upxi = q_0+ q_1\ii +q_2\jj + q_3\kk,
\end{equation}
for some $q_0,q_1,q_2,q_3\in\QQ_p$. It is further possible to associate with every quaternion $\upxi$ in $\HH_p$ a suitable matrix in $\mM(2,\QQ_{p,v})$, i.e., a $2\times 2$ matrix over the quadratic extension $\QQ_{p,v}\coloneq\QQ_p(\sqrt{v})$ of the field $\QQ_p$. To prove this, one first notices that $\HH_p$ has a natural structure of a $\QQ_{p,v}$-right vector space of dimension $2$, with basis $\{\Id,\jj\}$; that is, $\HH_p\simeq \QQ_{p,v}\oplus \jj \ \QQ_{p,v}$, and every $\upxi\in\HH_p$ can be expressed as 
\begin{equation}
\upxi = q_0+q_1\ii + q_2\jj +q_3\kk = (q_0+q_1\ii)\Id+ \jj(q_2-q_3\ii). 
\end{equation}
Next, consider the \emph{left-regular representation} $\mathcal{L}\colon \HH_p\rightarrow \End(\HH_p)$ ---given by left multiplication in $\HH_p$--- which associates with every quaternion $\upxi$ in $\HH_p$ a $\QQ_p$-linear endomorphism in $\End(\HH_p)$. In the basis $\{\Id, \jj\}$, we have that $\End(\HH_p)\simeq \mM(2,\QQ_{p,v})$. It is then clear that $\mathcal{L}$ provides an injective $\QQ_p$-algebra homomorphism between $\HH_p$ and $\mM(2,\QQ_{p,v})$, which is an isomorphism to its image $\mathcal{L}(\HH_p)\eqqcolon \Hi_p\subset\mM(2,\QQ_{p,v})$. In particular, observing that $\mathcal{L}$ acts on the basis elements $\ii$ and $\jj$ of $\HH_p$ as
\begin{equation}
\mathcal{L}(\ii) = 
\begin{pmatrix}
\sqrt{v} & 0 \\
0 &  -\sqrt{v}
\end{pmatrix},\quad
\mathcal{L}(\jj)=
\begin{pmatrix}
0 & -p\\
1 & 0
\end{pmatrix},
\end{equation}
we can conclude that the general form of an element $M$ in $\Hi_p$ ---and, hence, of the matrix associated with the quaternion $\upxi=q_0+q_1\ii + q_2\jj + q_3\kk$--- is given by 
\begin{equation}
M=
\begin{pmatrix}
q_0+q_1\sqrt{v} & -p(q_2+q_3\sqrt{v})\\
q_2-q_3\sqrt{v} & q_0-q_1\sqrt{v} 
\end{pmatrix}.
\end{equation}
To proceed further in our investigation of $p$-adic quaternion algebras, we now endow $\HH_p$ with a suitable \emph{involution} operation. In its most general setting, an involution is defined over an abstract $\F$-algebra $\Al$ as follows
\begin{definition}\label{def.3.3}
An \emph{involution} is a map $\overline{(\hspace{0.12mm}\cdot\hspace{0.12mm})}\colon \Al\rightarrow \Al$, over an $\F$-algebra $\Al$ into itself, which preserves the identity in $\Al$, i.e., $\overline{\Id}=\Id$, has degree two, namely, $\overline{\overline{\alpha}}=\alpha$, $\forall \alpha\in\Al$, and is an anti-automorphism of $\Al$, that is, $\overline{\alpha\beta}=\overline{\beta}\overline{\alpha}$, $\forall \alpha,\beta\in\Al$. We say that the involution is \emph{standard} if $\overline{\alpha}\alpha\in\F$ for all $\alpha$ in $\Al$.
\end{definition}
\begin{remark}
It is well known that an $\F$-algebra $\Al$, with $\ch(\F)\neq 2$, admits a standard involution if and only if $\Al$ has degree at most $2$ --- where, we recall that the degree of an algebra is the smallest integer $m\in\N$ such that every element $\alpha\in\Al$ satisfies a monic polynomial $f(x)\in\F[x]$ of degree $m$ (if no such integer $m$ exists one says that $\Al$ has degree $\infty$). Moreover, a division $\F$-algebra is non-commutative and of degree $2$ if and only if it is a quaternion algebra (see Corollary~$3.5.6$ in~\cite{quat1}). Therefore, we see that
among the non-commutative $\F$-algebras, the existence of a standard involution characterizes precisely the quaternion algebras: If $\Al$ is a non-commutative $\F$-algebra admitting a standard involution --- which is necessarily unique --- then it is a quaternion algebra of degree 2; conversely, every division quaternion algebra over $\F$ has degree at most $2$, and has a standard involution.
\end{remark}

Specializing Definition~\ref{def.3.3} to the $p$-adic quaternion algebra $\HH_p$, one is naturally led to the map
\begin{equation}
\overline{(\hspace{0.12mm}\cdot\hspace{0.12mm})}\colon \HH_p\ni\upxi= q_0+q_1\ii+q_2\jj+q_3\kk \mapsto \overline{\upxi}=q_0-q_1\ii-q_2\jj-q_3\kk\in\HH_p,
\end{equation}
which provides the (necessarily unique) standard involution over $\HH_p$. Once a standard involution is defined, we can further introduce two other important notions, namely, that of the \emph{reduced norm} and \emph{reduced trace} on $\HH_p$ (these maps will play an important role when relating $p$-adic quaternions with rotations).
\begin{definition}
Let $p>2$ be a prime number. The \emph{reduced norm} $\Nrd$ on $\HH_p$, is the map defined, for every quaternion $\upxi=q_0+q_1\ii+q_2\jj+q_3\kk$, as
\begin{equation}\label{eq.16}
\HH_p\ni\upxi\mapsto\Nrd(\upxi)\coloneqq\overline{\upxi}\upxi= q_0^2-vq_1^2+pq_2^2-vpq_3^2\in\QQ_p. 
\end{equation}
The \emph{reduced trace}, $\Trd$, on $\HH_p$ is defined as the map
\begin{equation}\label{eq.17}
\HH_p\ni\upxi\mapsto\Trd(\upxi)\coloneqq \upxi+\overline{\upxi} = 2q_0 \in\QQ_p.
\end{equation}
\end{definition}
From Eq.\eqref{eq.17} it is clear that the reduced trace is a $\QQ_p$-linear map, i.e., $\Trd(\alpha\upxi+\zeta)=\alpha\Trd(\upxi)+\Trd(\zeta)$, for every $\alpha\in\QQ_p$ and $\upxi,\zeta\in\HH_p$. It can be also easily checked that $\Nrd$ is a multiplicative map, namely, $\Nrd(\upxi\zeta)=\Nrd(\upxi)\Nrd(\zeta)$, $\upxi,\zeta\in\HH_p$. Moreover, from Eq.\eqref{eq.16}, we see that $\Nrd(\upxi)=Q_+^{(4)}(q_0,q_1,q_2,q_3)$, where $Q_+^{(4)}(\bm{x})$ is the quadratic form on $\QQ_p^4$ as in Eq.\eqref{eq:quadrform3}. Since this quadratic form is non-isotropic, we have that $\Nrd(\upxi)\neq 0$ for every $\upxi\in\HH_p^\times$; this entails that every $\upxi\in\HH_p^\times$ admits a two-sided inverse (i.e., $\upxi$ is a unit in $\HH_p$, see Lemma~$3.3.6$ in~\cite{quat1}), given by
\begin{equation}
\upxi^{-1}= \frac{1}{\Nrd(\upxi)}\overline{\upxi}.
\end{equation}
\begin{remark}
We have seen that the reduced norm of $\HH_p$ corresponds to the unique non-isotropic rank-$4$ quadratic form on $\QQ_p$. Thus, $\HH_p$ is indeed a division algebra.
By Theorem 12.3.2 in \cite{quat1}, when $\mathbb{F}\neq\mathbb{C}$ is a local field, there is a unique division quaternion algebra over $\mathbb{F}$ up to isomorphisms (by Theorem 2.5 in \cite{quat2}, two quaternion algebras over $\mathbb{F}$ are isomorphic if and only if their associated quadratic forms are equivalent). Therefore $\HH_p$ is the unique (division) quaternion algebra over $\QQ_p$ up to isomorphisms.
\end{remark}
Using the reduced trace and the reduced norm, we can single out two important subsets in $\HH_p$. In particular, we denote by $\HH_p^0$ the set defined as
\begin{equation}
    \HH_p^0\coloneqq \{\upxi\in\HH_p\mid\Trd(\upxi)=0\},
\end{equation}
namely the subset of quaternions in $\HH_p$ with zero trace, while we denote with
\begin{equation}
\HH_p^1\coloneqq\{\upxi\in\HH_p^\times\mid \Nrd(\upxi)=1\},
\end{equation}
the subgroup of quaternions in $\HH_p^\times$ with reduced norm equal to $1$. One can easily check that $\mathbb{H}_p^1$ is a normal subgroup of the multiplicative group $\mathbb{H}_p^\times$ of non-zero quaternions, while $\HH_p^0\simeq\{0\}\times\QQ_p^3\simeq\QQ_p^3$ is a subspace of $\HH_p$ consisting of \emph{purely imaginary quaternions} (i.e., of those quaternions $\upxi$ with $q_0=0$). We will shed further light on these subspaces in the next section, devoted to the relations between $p$-adic quaternions and rotations in $\QQ_p^3$.

\section{Relating $p$-adic quaternions and $p$-adic rotations}\label{sec.p-adic_quat_rot}
We now expose the relation between the $p$-adic quaternions $\HH_p$ and $\SO(3)_p$.
First, the group $\HH_p^\times$ acts on $\HH_p$ by \emph{conjugation}, i.e. 
\beq \label{eq:conjACT}
\HH_p\ni \upeta\mapsto  \upxi\upeta\upxi^{-1}\in\HH_p,
\eeq 
for all $\upeta\in\HH_p$, $\upxi\in \HH_p^\times$. This map is an \emph{isometric linear transformation} of $\HH_p$ since it preserves the reduced norm of every $\upeta\in \HH_p$. Moreover, the operation $\upeta\mapsto\upxi\upeta\upxi^{-1}$ leaves the centre $\QQ_p$ of $\HH_p$ pointwise fixed and hence also leaves the orthogonal subspace $\QQ_p^3$ invariant. Let $\kappa_p(\upxi)$ denote the restriction of the conjugation action~\eqref{eq:conjACT} to the subset $\HH_p^0$, i.e.,
\beq \label{eq:kappapxidef}
\HH_p^\times\circlearrowright \HH_p^0\rightarrow \HH_p^0,\qquad \kappa_p(\upxi)(\upeta)\coloneqq \upxi\upeta\upxi^{-1},
\eeq 
for every $\upeta\in\HH_p^0$ and $\upxi\in\HH_p^\times$. Since the action~\eqref{eq:kappapxidef} is an isometric transformation of $\HH_p$, $\kappa_p(\upxi)$ preserves the restriction of $\mathrm{nrd}\sim Q_+^{(4)}$ to $\mathbb{H}_p^0\simeq \mathbb{Q}_p^3$ --- which is $\left.\mathrm{nrd}\right\rvert_{\mathbb{H}_p^0}\sim Q_+$--- i.e., $\kappa_p(\upxi)$ is an orthogonal transformation in $\mathrm{O}(\left.\rn\right\rvert_{\HH_p^0})$. Furthermore, one can write $\kappa_p(\upxi)$ as a $3\times 3$ matrix with respect to the basis $(\mathbf{i},\mathbf{j},\mathbf{k})$ of $\QQ_p^3$. We explicitly derive the action of an invertible quaternion on a pure imaginary quaternion: if $\upxi=q_0+\mathbf{i} q_1+\mathbf{j} q_2+\mathbf{k} q_3\in\HH_p^\times$ and $\upeta=\mathbf{i} s_1+\mathbf{j} s_2+\mathbf{k} s_3\in\HH_p^0$, 
\beq\label{kpxi}
\kappa_p(\upxi)(\upeta)=(q_0+\mathbf{i} q_1+\mathbf{j} q_2+\mathbf{k} q_3)(\mathbf{i} s_1+\mathbf{j} s_2+\mathbf{k} s_3)\frac{q_0-\mathbf{i} q_1-\mathbf{j} q_2-\mathbf{k} q_3}{\rn(\upxi)}.
\eeq
Expanding this product, the scalar part vanishes as expected, and, by collecting the terms in $\mathbf{i},\mathbf{j}$ and $\mathbf{k}$, 
we get the following matrix representation $K_p(\upxi)$ of the map $\kappa_p(\upxi)(\upeta)$
{\small\beq\label{eq:matrkxi}
K_p(\upxi)=\frac{1}{\rn(\upxi)}
\begin{pmatrix}
q_0^2-vq_1^2-pq_2^2+vpq_3^2 & 2p(q_1q_2-q_0q_3) & 2p(q_0q_2-vq_1q_3)\\
-2v(q_0q_3+q_1q_2) & q_0^2+vq_1^2+pq_2^2+vp q_3^2 & 2v(q_0q_1-pq_2q_3)\\
-2(q_0q_2+vq_1q_3) & 2(q_0q_1+pq_2q_3) & q_0^2+vq_1^2-pq_2^2-vpq_3^2
\end{pmatrix},
\eeq}for $p>2$, with respect to the basis $(\mathbf{i},\mathbf{j},\mathbf{k})$ of $\QQ_p^3$. From this matrix representation, it is easy to see that $\det K_p(\upxi)=1$ and hence $\kappa_p(\upxi)\in\SO(\left.\rn\right\rvert_{\HH_p^0})$. Therefore, from the action \eqref{kpxi}, this defines a map $\kappa_p$ that maps elements in $\HH_p^{\times}$ to rotations in $\SO(\left.\rn\right\rvert_{\HH_p^0}),$ i.e.,
\beq 
\kappa_p\colon \HH_p^\times \rightarrow \SO(\left.\rn\right\rvert_{\HH_p^0}).
\eeq 
Moreover, since the quadratic form $Q_+$ defining $\SO(3)_p$ is equivalent to $\left.\rn\right\rvert_{\HH_p^0},$ we have
\begin{equation}\label{isoSOgroups}
\SO(\left.\rn\right\rvert_{\HH_p^0})\simeq \SO(3)_p,
\end{equation} 
where $\left.\rn\right\rvert_{\HH_p^0}(\upxi)=-vq_1^2+pq_2^2-vpq_3^2$.

From Proposition 4.5.10 in \cite{quat1} and Theorem 3.1 p. 63 in \cite{quat2}, the action \eqref{eq:kappapxidef} induces a short exact sequence of groups:
\beq\label{eq:shortexact}
1\rightarrow \QQ_p^\times \hookrightarrow \HH_p^\times\ {\stackrel{\textup{\upshape $\kappa_p$}}{\twoheadrightarrow}}\ \SO(\left.\rn\right\rvert_{\HH_p^0})\rightarrow 1,
\eeq
and, from Theorem 4.14 and Proposition 4.15 in \cite{our3rd}, this entails that the following is an isomorphism of topological groups:
\beq \label{eq:isoSO3Hpsc}
\begin{split}
    \uppsi_p\colon \HH_p^\times/\QQ_p^\times&\rightarrow \SO(\left.\rn\right\rvert_{\HH_p^0})\simeq \SO(3)_p\\ 
    \upxi\QQ_p^\times&\mapsto \kappa_p(\upxi).
\end{split}
\eeq 

Note that $\SO(3)_p\simeq \mathbb{H}_p^\times / \mathbb{Q}_p^\times$ is formally different from $\SO(3)_\RR\simeq \mathbb{H}^1/\{\pm1\}$ over $\mathbb{R}$. Given the isomorphism \eqref{eq:isoSO3Hpsc}, we want to find a relation between quadruples $(q_0,q_1,q_2,q_3)\in\QQ_p^4/\QQ_p^\times$ of quaternion coefficients modulo scalars and nautical angles $(\alpha,\beta,\gamma)\in (\QQ_p\cup\{\infty\})^3$.

\medskip

According to the nautical angle decompositions of $\SO(3)_p$ in Eq.\eqref{eq:Cardanodecos}, for every prime $p>2$, every rotation can be written as a product of rotations around the axes with respect to the canonical basis of $\QQ_p^3$ as follows:
\begin{align}
&\SO(3)_p\ni \cR=\cR_z(\alpha)\cR_y(\beta)\cR_x(\gamma)\nonumber\\
& = \begin{pmatrix}
\frac{1+v\alpha^2}{1-v\alpha^2} &
\frac{2v\alpha}{1-v\alpha^2}&0\\
\frac{2\alpha}{1-v\alpha^2} & \frac{1+v\alpha^2}{1-v\alpha^2} &0\\0&0&1
\end{pmatrix}
\begin{pmatrix}
\frac{1-p\beta^2}{1+p\beta^2} & 0 & -\frac{2p\beta}{1+p\beta^2}\\
0&1&0\\
\frac{2\beta}{1+p\beta^2} & 0&\frac{1-p\beta^2}{1+p\beta^2}
\end{pmatrix}
\begin{pmatrix}
    1&0&0\\
0&\frac{1+\frac{p}{v}\gamma^2}{1-\frac{p}{v}\gamma^2} &
\frac{2\frac{p}{v}\gamma}{1-\frac{p}{v}\gamma^2}\\
0&\frac{2\gamma}{1-\frac{p}{v}\gamma^2} & \frac{1+\frac{p}{v}\gamma^2}{1-\frac{p}{v}\gamma^2}
\end{pmatrix}= 
\begin{pmatrix}
\cR_{11}&\cR_{12}&\cR_{13}\\\cR_{21}&\cR_{22}&\cR_{23}\\\cR_{31}&\cR_{32}&\cR_{33}
\end{pmatrix},
\label{cardanorotaz}
\end{align}
where
\begin{align}
    &\cR_{11}=\frac{1 + v \alpha^2}{(1 - v \alpha^2)} \frac{1 - p \beta^2}{1 + p \beta^2},\\
    &\cR_{12}=\frac{2 v \alpha (1 + \frac{p}{v}\gamma^2)}{(1 - v \alpha^2)(1 - \frac{p}{v}\gamma^2)} -\frac{4 p (1 + v \alpha^2) \beta \gamma}{(1 - v \alpha^2)(1 + p \beta^2)(1 - \frac{p}{v}\gamma^2)},\\
    &\cR_{13}=\frac{4 p \alpha \gamma}{(1 - v \alpha^2)(1 - \frac{p}{v}\gamma^2)} - \frac{2 p (1 + v \alpha^2) \beta (1 + \frac{p}{v}\gamma^2)}{(1 - v \alpha^2)(1 + p \beta^2)(1 - \frac{p}{v}\gamma^2)},\\
    &\cR_{21}=\frac{2 \alpha (1 - p \beta^2)}{(1 - v \alpha^2)(1 + p \beta^2)}, \\
    &\cR_{22}=\frac{1 + v \alpha^2}{1 - v \alpha^2}\frac{1 + \frac{p}{v}\gamma^2}{1 - \frac{p}{v}\gamma^2} -\frac{8 p \alpha \beta \gamma}{(1 - v \alpha^2)(1 + p \beta^2)(1 - \frac{p}{v}\gamma^2)},\\
    &\cR_{23}=\frac{2 p (1 + v \alpha^2) \gamma}{v (1 - v \alpha^2)(1 - \frac{p}{v}\gamma^2)} - \frac{4 p \alpha \beta (1 + \frac{p}{v}\gamma^2)}{(1 - v \alpha^2)(1 + p \beta^2)(1 - \frac{p}{v}\gamma^2)},\\
    &\cR_{31}=\frac{2 \beta}{1 + p \beta^2} ,\\
    &\cR_{32}=\frac{2 (1 - p \beta^2) \gamma}{(1 + p \beta^2)(1 - \frac{p}{v}\gamma^2)}, \\
    &\cR_{33}=\frac{1 - p \beta^2}{1 + p \beta^2} \frac{1 + \frac{p}{v}\gamma^2}{1 - \frac{p}{v}\gamma^2}.
\end{align}

Consider now $\kappa_p(\upxi)\in\SO(\left.\rn\right\rvert_{\HH_p^0})$ in its matrix form $K_p(\upxi)$ as in Eq.\eqref{eq:matrkxi}. Let $\left.A_+^{(4)}\right\rvert_{\HH_p^0}\coloneqq\diag(-v,p,-vp)$ be the matrix representation of the quadratic form $\left.\rn\right\rvert_{\HH_p^0}$ w.r.t. the quaternionic basis $\mathbf{i},\mathbf{j},\mathbf{k}$. Then we have that
\beq
K_p(\upxi)^\top\left.A_+^{(4)}\right\rvert_{\HH_p^0}K_p(\upxi)=\left.A_+^{(4)}\right\rvert_{\HH_p^0}, 
\eeq
whereas for $A_+$
\beq
K_p(\upxi)^\top A_+K_p(\upxi)\neq A_+,
\eeq 
which means that $K_p(\upxi)\in \SO(A_+^{(4)}|_{\mathbb{H}_p^0})$ and $K_p(\upxi)\not\in \SO(3)_p$. Since $\SO(A_+^{(4)}|_{\mathbb{H}_p^0})$ and $\SO(3)_p$ are isomorphic due to the equivalence of their respective quadratic forms, we have to find a change of basis in order to obtain the right matrix expression of $\kappa_p(\upxi)$ living in $\SO(3)_p$, to be equated to Eq.\eqref{cardanorotaz}.

Let us therefore deepen this equivalence of quadratic forms, to find an explicit isomorphism for Eq.\eqref{isoSOgroups}. Actually we have the following

\begin{theorem}\label{finaliso}
For every prime $p>2$, the isomorphism of topological groups between $\SO(3)_p$ and $\HH_p^\times/\QQ_p^\times$ in Eq.\eqref{eq:isoSO3Hpsc} is
\beq \label{totalcorisoQpHpSO3}
\mathcal{T}_p\colon \HH_p^\times/\QQ_p^\times\rightarrow  \SO(3)_p, \quad
\upxi\QQ_p^\times\mapsto\Lambda K_p(\upxi)\Lambda^{-1},
\eeq 
where $\Lambda\coloneqq \begin{pmatrix} 0&0&-vp\\0&p&0\\-v&0&0 \end{pmatrix}$ and 
\beq \label{eq:matrkxi2}
\Lambda K_p(\upxi)\Lambda^{-1}=\frac{1}{\rn(\upxi)}
\begin{pmatrix}
q_0^2+vq_1^2-pq_2^2-vpq_3^2 & -2v(q_0q_1+pq_2q_3) & -2p(q_0q_2+vq_1q_3)\\
2(pq_2q_3-q_0q_1) & q_0^2+vq_1^2+pq_2^2+vp q_3^2 & 2p(q_1q_2+q_0q_3)\\
2(q_0q_2-vq_1q_3) & 2v(q_0q_3-q_1q_2) & q_0^2-vq_1^2-pq_2^2+vpq_3^2
\end{pmatrix},
\end{equation}
for all $\upxi\in\HH_p^\times$.
\end{theorem}
\begin{proof}
    First, we observe that 
\beq 
-vp\left.A_+^{(4)}\right\rvert_{\HH_p^0} \eqqcolon A_s = \diag\big((-v)^2p,-vp^2,(-vp)^2\big),
\eeq 
which yields $\SO\big(\left.A_+^{(4)}\right\rvert_{\HH_p^0}\big)=\SO(A_s)$ since, for any $L\in \mM(3,\QQ_p)$, $L^\top A_s L = A_s$ if and only if $L^\top \left.A_+^{(4)}\right\rvert_{\HH_p^0} L =\left.A_+^{(4)}\right\rvert_{\HH_p^0}$. Moreover, the square factors of the coefficients of a quadratic form can be absorbed in $\QQ_p^{\times^2}$ and, swapping the first component with the third, we get
\beq
A_s\mapsto\diag(p,-v,1)\mapsto A_+.
\eeq 
This is indeed a similarity of quadratic forms leading to isomorphic (special) orthogonal groups that are realized by the matrix $\Lambda$ (related to a change of basis of $\QQ_p^3$) since $\Lambda^\top A_+\Lambda= A_s.$
Therefore, we have the following isomorphism of (topological) groups using the conjugation matrix $\Lambda$:
\beq 
\SO(3)_p\ni\cR\ {\stackrel{\textup{\upshape $1\hspace{-0,1cm}:\hspace{-0,1cm}1$}}{\longleftrightarrow}}\ \Lambda^{-1}\cR\Lambda\in \SO(A_s)=\SO\big(\left.A_+^{(4)}\right\rvert_{\HH_p^0}\big),
\eeq 
which, after basic matrix calculations, yields the result.
\end{proof}

\begin{corollary}\label{q0solutions}
Equating \eqref{cardanorotaz} and \eqref{eq:matrkxi2}, we obtain the following solutions depending on the free parameter $q_0\in\QQ_p^\times$:
\beq \label{mainrelparq0}
q_1 = \frac{\frac{p}{v}\beta\gamma-\alpha}{1-p\alpha\beta\gamma}q_0,\qquad q_2=\frac{\beta-\alpha\gamma}{1-p\alpha\beta\gamma}q_0,\qquad q_3=\frac{\frac{\gamma}{v}-\alpha\beta}{1-p\alpha\beta\gamma}q_0.
\eeq 
\end{corollary}

\begin{remark}\label{rem:projectiverel}

In Eq.\eqref{mainrelparq0} there is an apparent singularity when $1 - p\alpha\beta\gamma = 0$. Since isomorphism \eqref{totalcorisoQpHpSO3} is defined on the quotient group $\HH_p^\times/\QQ_p^\times$, a non-zero quaternion $\upxi=q_0+\mathbf{i} q_1+\mathbf{j}q_2+\mathbf{k}q_3$ determines a rotation up to a non-zero scalar multiple. Thus, the natural parameter space for these rotations is the projective space $\mathrm{P}^3(\QQ_p)$, equipped with homogeneous coordinates $[q_0 : q_1 : q_2 : q_3]$. By deriving Eq.\eqref{mainrelparq0} and expressing $q_1, q_2$, $q_3$ in terms of $q_0$, we are implicitly restricting our analysis to the affine chart defined by $q_0 \neq 0$. By rewriting Eq.\eqref{mainrelparq0} in terms of projective coordinates, we obtain the global parametrization \begin{equation}\notag [q_0 : q_1 : q_2 : q_3] = \left[1-p\alpha\beta\gamma:\frac{p}{v}\beta\gamma-\alpha:\beta-\alpha\gamma:\frac{\gamma}{v}-\alpha\beta\right]. \end{equation} Indeed, this is globally well-defined and when restricted to $1-p\alpha\beta\gamma=0$ is not the trivial class $[0:0:0:0]$. To see this, from the first component equal to zero we get $\beta\gamma=\frac{1}{p\alpha}$ ($\alpha\neq0$ otherwise $1-p\alpha\beta\gamma=0\Rightarrow 1=0$), and from the second one equal to zero we get $\beta\gamma=\frac{v\alpha}{p}$; this implies $\frac{1}{p\alpha}=\frac{v\alpha}{p}$, i.e., $v=\frac{1}{\alpha^2}$, which is a contradiction for every $\alpha\neq0$ since $v$ is not a square. From the point of view of projective geometry, the condition $1 - p\alpha\beta\gamma = 0$ does not represent a mathematical singularity, but rather the vanishing of the scalar component $q_0 = 0$. A quaternion with zero scalar component is a pure imaginary quaternion in $\HH_p^0$. One can show (see e.g. Section 4.5 in \cite{quat1}) that $\upxi\in \HH_p^0$ if and only if the corresponding matrix %\Lambda K_p(\upxi)\Lambda^{-1}
in $\SO(3)_p$ is an involution (i.e., a rotation $R\neq \mathrm{I}_3$ such that $R^2=\mathrm{I}_3$), corresponding to a rotation of parameter $\infty$ around the axis associated to $\upxi$. Therefore, $\mathbb{H}_p^0$ corresponds to $D\coloneqq\{(\alpha,\beta,\gamma)\in (\QQ_p\cup\{\infty\})^3\ \midd\ 1-p\alpha\beta\gamma=0\}$ which, in turn, corresponds to involutions.
Moreover, $D$ has zero Haar measure since is an algebraic manifold in $\QQ_p^3$ with codimension equal to $1$.
\end{remark}

We could have found relations between quaternion components and nautical angles in terms of the other free parameters $q_1,q_2,q_3$, but the choice of parameter $q_0$ remains the preferred one due to the fact that Eq.\eqref{mainrelparq0} encompasses the rotations around the reference axes, since they are such that $1-p\alpha\beta\gamma=1\neq0$. Indeed, we can find the formulas for the rotations around the reference axes equating rotations $\cR_z(\alpha),\cR_y(\beta),\cR_x(\gamma)$ with $\Lambda K_p(\upxi)\Lambda^{-1}$ in \eqref{eq:matrkxi2} which yields the following formulas:
\begin{enumerate}
    \item[$\bullet$] for a rotation $\cR_z(\alpha)$ around the $z$-axis we have
    \beq \label{convform3}
    q_1=-\alpha q_0,\qquad q_2=q_3=0,
    \eeq
    \item[$\bullet$] for a rotation $\cR_y(\beta)$ around the $y$-axis (i.e. $\alpha=\gamma=0$) we have
    \beq \label{convform2}
    q_2=\beta q_0,\qquad q_1=q_3=0,
    \eeq 
    \item[$\bullet$] for a rotation $\cR_x(\gamma)$ around the $x$-axis (i.e. $\alpha=\beta=0$) we have
    \beq\label{convform4}
    q_3=\frac{\gamma}{v}q_0,\qquad q_1=q_2=0.
    \eeq
\end{enumerate}

\section{Haar measure}\label{sec:Haarfoundations}

We now recall the Haar integrals on the bidimensional and tridimensional rotation groups $\SO(2)_{p,d}$ and $\SO(3)_{p}.$ For comprehensive and explicit construction of the left (and right) Haar measure $\mu$ on a (second countable) $p$-adic Lie group $G$ we refer to Section 3 in \cite{our3rd}. 
According to the parametrization for $\SO(2)_{p,d}$ present in Theorem 2.38 in \cite{our3rd} we have that the group is homeomorphic to the $p$-adic projective line and it can be covered by two distinct charts. Applying the general formula for the Haar measure on $p$-adic Lie groups (see e.g. Theorem 3.2 in \cite{our3rd}) on $\SO(2)_{p,d}$ yields the following Haar measure:
\beq\label{so2measure}
\mu_{\SO(2)_{p,d}}(E)=\int_{\varphi_d(E)}\frac{1}{|{1+d\sigma^2}|_{p}} \ d\lambda(\sigma),
\eeq
where $\varphi_d(E)$ is a coordinate map with $\varphi_d(I)=0,$ for every Borel subset $E$ of $\SO(2)_{p,d}$ and where $\lambda$ is the Haar measure on $\QQ_p$. Note that the factors at the denominators are coming with $d\in\{-v,p,-p/v\}$, such that $-d$ is not a square, hence $1+d\sigma^2\neq0$ for all $\sigma\in\mathbb{Q}_p$. From Theorem~5.40 in $\cite{thesis}$ we have the following normalization factors:
    \beq\label{so2normfact}
    \begin{split}
        &\mu_{\SO(2)_{p,-v}}(\SO(2)_{p,-v})=1+\frac{1}{p},\\
        &\mu_{\SO(2)_{p,p}}(\SO(2)_{p,p})=2,\\
        &\mu_{\SO(2)_{p,-p/v}}(\SO(2)_{p,-p/v})=2.
    \end{split}
    \eeq

The construction we will use for the Haar measure on the compact $p$-adic Lie group $\SO(3)_p$ relies on Theorem~\ref{finaliso} and Corollary~\ref{q0solutions}. 
To begin with, the group $\HH_p^\times$ admits a left Haar measure, as it is a locally compact group; moreover, the group $\HH_p^\times$ is unimodular for every prime, $p$ which means that the left and right Haar measures coincide and can be constructed by exploiting directly the general formula present in Theorem 3.2 in \cite{our3rd} therefore obtaining:
\beq \label{Hmeasure}
\mu_{\HH_p^\times}(E)=\int_{\varphi(E)}\frac{1}{\big\lvert Q_+^{(4)}(\mathbf{q})\big\rvert_p^2}\ d\lambda(\mathbf{q}),
\eeq
for every Borel subset $E$ where $\upxi=q_0+\mathbf{i} q_1+\mathbf{j} q_2+\mathbf{k} q_3\in\HH_p^\times$, $\varphi(\upxi)=(q_0,q_1,q_2,q_3)=\mathbf{q},$ $d\lambda(\mathbf{q})=d q_0d q_1d q_2d q_3$ is the Haar measure on $\QQ_p^4$ and $Q_+^{(4)}$ is the non-degenerate non-isotropic quadratic form on $\mathbb{Q}_p^4$ in Eq.\eqref{eq:quadrform3}.

Finally, we are going to state and prove our main result, i.e., the expression of the Haar measure on $\SO(3)_p$ in terms of nautical angles.

\begin{theorem}\label{finalhaarmeasuret}
For every prime $p>2$ the Haar measure $d\mu_{\SO(3)_p}(\cR)$ on $\SO(3)_p$ in terms of nautical angles $\alpha,\beta,\gamma\in \QQ_p\cup \infty$ such that $\mathcal{R}(\alpha,\beta,\gamma)=\mathcal{R}_z(\alpha)\mathcal{R}_y(\beta)\mathcal{R}_x(\gamma)$ is given by 
\beq\label{finalhaarmeasure}
\int_{\mathcal{B}}d\mu_{\SO(3)_p}(\cR(\alpha,\beta,\gamma)) = \left[4\left(1+\frac{1}{p}\right)\right]^{-1}\int_{\mathcal{B}}\frac{d\alpha d\beta d\gamma}{\left\lvert (1-v\alpha^{2})(1+p\beta^{2})(1-\frac{p}{v}\gamma^2) \right\rvert_p},
\eeq 
for every Borel subset $\mathcal{B}\in\mathscr{B}(\SO(3)_p)$\footnote{It should be clear that, with a slight abuse of notation, on the r.h.s.\ of Eq.~\eqref{finalhaarmeasure} we have used the same symbol $\mathcal{B}$ to denote the Borel subset of $\QQ_p^3$ corresponding --- via the nautical angle parametrization of $\SO(3)_p$ --- to the Borel set $\mathcal{B}\in\mathscr{B}(\SO(3)_p)$.}, and
where $v\in \QQ_p$ is a non-quadratic unit as in Eq.\eqref{v}.
\end{theorem}

\begin{proof}
To find the Haar measure on $\SO(3)_p$ in terms of nautical angles, we first exploit the group isomorphism $\mathcal{T}_p$ in Eq.\eqref{totalcorisoQpHpSO3}, which allows us to identify the Haar measure and integral on $\SO(3)_p$ with those on $\HH_p^\times/\QQ_p^\times$. Our strategy is to deduce the measure on the quotient $\HH_p^\times/\QQ_p^\times$ starting from the known Haar measure on the quaternion group $\HH_p^\times$, given in Eq.\eqref{Hmeasure}. To this end, we exploit Theorem 2.51 in \cite{folland2016course} which relates the respective measures. We restrict our derivation to the affine chart where $q_0 \neq 0$, since the complementary locus $q_0 = 0$ has zero Haar measure. Indeed, once the measure on $\HH_p^\times/\QQ_p^\times$ is explicitly determined on this chart, we employ the relations of Eq.\eqref{mainrelparq0} defined for $q_0\neq0$, to perform a change of variables from the quaternionic components to nautical angles.

Let us now define the subspace of quaternions with non-zero scalar part,
\beq 
\widetilde{U_0}\coloneqq\{\upxi= q_0+q_1\ii+q_2\jj+q_3\kk\in\HH_p^\times\ \midd\ q_0\neq0\}\subset \HH_p^\times, 
\eeq 
isomorphic to the topological space $\QQ_p^\times\times\QQ_p^3$ through $\upxi\leftrightarrow(q_0,q_1,q_2,q_3)$. Since $q_0\neq0$, every $\upxi \in \widetilde{U_0}$ is of the form $\upxi= q_0+q_1\ii+q_2\jj+q_3\kk = q_0\left(1+\frac{q_1}{q_0}\ii+\frac{q_2}{q_0}\jj+\frac{q_3}{q_0}\kk\right)$, i.e., it admits a unique factorization $\upxi = q_0 \widetilde{x}$, where $q_0 \in \QQ_p^\times$ and $\widetilde{x}\coloneqq 1+x_1\ii +x_2\jj +x_3\kk\in\widetilde{U_0},$ having set $x_i\coloneqq\frac{q_i}{q_0}$, $i=1,2,3$
%, $(x_1, x_2, x_3)\in \QQ_p^3$
. The quotient 
\beq 
U_0\coloneqq \widetilde{U_0}/\QQ_p^\times \subset \HH_p^\times/\QQ_p^\times,
\eeq 
is isomorphic to the affine chart of homogeneous coordinates
$\{[q_0:q_1:q_2:q_3]\in\mathrm{P}^3(\QQ_p)\ \midd\ q_0\neq0\}$ through the mapping $\upxi\QQ_p^\times\leftrightarrow[q_0:q_1:q_2:q_3]$. Again, we can write 
\beq
[q_0:q_1:q_2:q_3] = \left[1:\frac{q_1}{q_0}:\frac{q_2}{q_0}:\frac{q_3}{q_0}\right] = [1:x_1:x_2:x_3]
\eeq
and therefore 
$\widetilde{x}\QQ_p^\times\leftrightarrow[1:x_1:x_2:x_3]$.

Actually, the map
\beq \label{eq:hoemoTphi0}
\varphi_0\colon \QQ_p^3\to U_0,\quad (x_1,x_2,x_3)\mapsto %[1:x_1:x_2:x_3]
\widetilde{x}\QQ_p^\times
\eeq 
is an homeomorphism. Indeed, $\varphi_0$ is clearly  surjective and it is injective because  
$\widetilde{x}\QQ_p^\times = \widetilde{y}\QQ_p^\times$ if and only if there exists $\lambda\in\QQ_p^\times$ such that $\widetilde{x}= \lambda\widetilde{y}$, which implies $\lambda=1$. Moreover, both $\varphi_0,\varphi_0^{-1}$ are continuous by construction.

We want to use Theorem 2.51 in \cite{folland2016course} to relate the Haar measure on the group $G=\HH_p^\times$ to its quotient by its closed central subgroup $H=\QQ_p^\times$. Then, for every $f\in C_c(\HH_p^\times)$ we have
\beq \label{eq:primointe}
\int_{\HH_p^\times} f(\upxi) \ d\mu_{\HH_p^\times}(\upxi) = \int_{\HH_p^\times/\QQ_p^\times} \left( \int_{\mathbb{Q}_p^\times} f(\upxi q_0) \, d\mu_{\QQ_p^\times}(q_0) \right) d\mu_{\HH_p^\times/\QQ_p^\times}(\upxi\QQ_p^\times).
\eeq
As discussed in Remark \ref{rem:projectiverel}, the hyperplane $\widetilde{V_0}\coloneqq\{\upxi\in \HH_p^\times\ \midd \ q_0=0\}=\HH_p^0\setminus\{\mathbf{0}\}$ has codimension $1$ in $\HH_p^\times$ and hence it has zero measure with respect to the Haar measure on $\HH_p^\times$. A similar argument holds for its projection $V_0\coloneqq \widetilde{V_0}/\QQ_p^\times$ in $\HH_p^\times/\QQ_p^\times$. Consequently, the Haar measures $\mu_{\HH_p^\times/\QQ_p^\times}$ and $\mu_{\HH_p^\times}$ are fully supported on $U_0 $ and $\widetilde{U_0}$ respectively. 
Then, Eq.\eqref{eq:primointe} becomes
\beq \label{eq:secondointe}
\int_{\widetilde{U_0}} f(\upxi) \ d\mu_{\HH_p^\times}(\upxi) = \int_{U_0} \left( \int_{\mathbb{Q}_p^\times} f(\widetilde{x} q_0) \ d\mu_{\QQ_p^\times}(q_0) \right) d\mu_{\HH_p^\times/\QQ_p^\times}(\widetilde{x}\QQ_p^\times).
\eeq 

We now want to express the measure on $\widetilde{U_0}$ using coordinates. Recall from Eq.\eqref{Hmeasure} that $d\mu_{\HH_p^\times}(\upxi) = \frac{dq_0dq_1dq_2dq_3}{\lvert Q_+^{(4)}(q_0,q_1,q_2,q_3)\rvert_p^2}$. Considering the change of variables $(q_0, q_1, q_2, q_3) \mapsto (q_0, x_1, x_2, x_3)$ on $\QQ_p^\times\times\QQ_p^3$ with $x_i = \frac{q_i}{q_0},$ we obtain that the determinant of its Jacobian is 
\beq 
\det\frac{\partial(q_0,q_1,q_2,q_3)}{\partial(q_0,x_1,x_2,x_3)} = \det\begin{pmatrix}
1 & 0 & 0 & 0 \\
x_1 & q_0 & 0 & 0 \\
x_2 & 0 & q_0 & 0 \\
x_3 & 0 & 0 & q_0
\end{pmatrix} =q_0^3.
\eeq
By the change of variables theorem we obtain $dq_0 dq_1 dq_2 dq_3 = \lvert q_0\rvert_p^3 dq_0 dx_1 dx_2 dx_3$. Furthermore, we can write $Q_+^{(4)}(q_0, q_1, q_2, q_3)=q_0^2Q_+^{(4)}(1, x_1, x_2, x_3)$ and then the Haar measure on $\HH_p$ factorizes as
\beq 
d\mu_{\HH_p^\times}(\upxi) = 
d\mu_{\QQ_p^\times}(q_0)\, d\nu(x_1,x_2,x_3),
\eeq 
where $d\mu_{\QQ_p^\times}(q_0) = \frac{dq_0}{\lvert q_0\rvert_p}$ is the Haar measure on the multiplicative group $\QQ_p^\times$ and
\beq \label{eq:dnusuQp3}
d\nu(x_1,x_2,x_3)\coloneqq \frac{dx_1dx_2dx_3}{ \lvert Q_+^{(4)}(1,x_1,x_2,x_3)\rvert_p^2}.
\eeq 
By Fubini's theorem, the l.h.s. of Eq.\eqref{eq:secondointe} becomes
\beq \label{eq:primomembro}
\int_{\widetilde{U_0}} f(\upxi) \ d\mu_{\HH_p^\times}(\upxi) 
 = \int_{\QQ_p^3} \left( \int_{\QQ_p^\times} f\left(
\widetilde{x}(x_1,x_2,x_3)q_0\right)\ d\mu_{\QQ_p^\times}(q_0)\right) d\nu(x_1,x_2,x_3).
\eeq 

Now we push the integral on $\QQ_p^3$ forward the quotient $U_0$ through the homeomorphism $\varphi_0 \colon \QQ_p^3 \to U_0$ by using the push-forward measure of the measure $\nu$ through $\varphi_0,$ that is ${\varphi_0}_\ast \nu(\mathcal{E})\coloneq\nu\circ \varphi_0^{-1}(\mathcal{E})$ for every Borel set $\mathcal{E}\subseteq U_0$. Then, for every $F\in C_c(U_0)$ we have 
\beq
\int_{U_0} F(\widetilde{x}\QQ_p^\times)\ d({\varphi_0}_\ast\nu)(\widetilde{x}\QQ_p^\times)=\int_{\QQ_p^3} F(\varphi_0(x_1,x_2,x_3))\ d\nu(x_1,x_2,x_3). 
\eeq
By using the fact that $F(\widetilde{x}\QQ_p^\times) = \int_{\QQ_p^\times} f\left(\widetilde{x}
q_0\right)d\mu_{\QQ_p^\times}(q_0)$ and, since in Eq.\eqref{eq:primomembro} we have 
\beq
\int_{\QQ_p^\times} f\left(\widetilde{x}(x_1,x_2,x_3)q_0\right)d\mu_{\QQ_p^\times}(q_0) = F(\varphi_0(x_1,x_2,x_3)),
\eeq
we obtain that Eq.\eqref{eq:primomembro} translates to:
\beq \label{eq:pushforwardLHS}
\int_{\widetilde{U_0}} f(\upxi) \ d\mu_{\HH_p^\times}(\upxi) = \int_{U_0} \left( \int_{\QQ_p^\times} f(\widetilde{x}q_0) \ d\mu_{\QQ_p^\times}(q_0) \right)d({\varphi_0}_\ast\nu)(\widetilde{x}\QQ_p^\times).
\eeq 
Comparing the right-hand sides of Eqs.\eqref{eq:secondointe} and \eqref{eq:pushforwardLHS} we observe that they are equal for every $f\in C_c(\HH_p^\times)$ and hence, by the uniqueness of the Haar measure, it implies
\beq \label{eq:arrivoqui2}
d\mu_{\HH_p^\times/\QQ_p^\times}\big\rvert_{U_0} = d({\varphi_0}_\ast\nu).
\eeq

Finally, let us remember that $D^{c}\coloneqq\{(\alpha,\beta,\gamma)\in (\QQ_p\cup\{\infty\})^3\ \midd\ 1-p\alpha\beta\gamma\neq0\}$. Exploiting Eq.\eqref{mainrelparq0} in $q_0=1$ according to Eq.\eqref{eq:dnusuQp3}, we provide the map 
\beq 
\ell\colon D^c\to \QQ_p^3,\qquad (\alpha,\beta,\gamma)\mapsto\left(x_1=\frac{\frac{p}{v}\beta\gamma-\alpha}{1-p\alpha\beta\gamma}, x_2=\frac{\beta-\alpha\gamma}{1-p\alpha\beta\gamma}, x_3=\frac{\frac{\gamma}{v}-\alpha\beta}{1-p\alpha\beta\gamma}\right),
\eeq 
and evaluate its Jacobian $J$ as
\beq
\begin{split}
    J&=\begin{pmatrix}
    \frac{\partial x_1}{\partial \alpha}&\frac{\partial x_1}{\partial \beta}&\frac{\partial x_1}{\partial \gamma}\\
    \frac{\partial x_2}{\partial \alpha}&\frac{\partial x_2}{\partial \beta}&\frac{\partial x_2}{\partial \gamma}\\
    \frac{\partial x_3}{\partial \alpha}&\frac{\partial x_3}{\partial \beta}&\frac{\partial x_3}{\partial \gamma}
\end{pmatrix}=\begin{pmatrix}
J_{11}&J_{12}&J_{13}\\J_{21}&J_{22}&J_{23}\\J_{31}&J_{32}&J_{33}
\end{pmatrix},
\end{split}
\eeq
where
\begin{align}
    &J_{11}=-\dfrac{p \beta \gamma \left(v \alpha - p\, \beta \gamma\right)}{v \left(-1 + p \alpha \beta \gamma\right)^{2}} + \dfrac{1}{-1 + p \alpha \beta \gamma},\\
    &J_{12}=-\dfrac{p \alpha \gamma \left(v \alpha - p\beta \gamma\right)}{v \left(-1 + p \alpha \beta \gamma\right)^{2}} - \dfrac{p\, \gamma}{v \left(-1 + p \alpha \beta \gamma\right)},\\
    &J_{13}=-\dfrac{p \alpha \beta \left(v \alpha - p \beta \gamma\right)}{v \left(-1 + p \alpha \beta \gamma\right)^{2}} - \dfrac{p\, \beta}{v \left(-1 + p \alpha \beta \gamma\right)},\\
    &J_{21}=\dfrac{p\beta \gamma \left(\beta - \alpha \gamma\right)}{\left(-1 + p \alpha \beta \gamma\right)^{2}} + \dfrac{\gamma}{-1 + p \alpha \beta \gamma},\\
    &J_{22}=\dfrac{p\alpha \gamma \left(\beta - \alpha \gamma\right)}{\left(-1 + p \alpha \beta \gamma\right)^{2}} - \dfrac{1}{-1 + p \alpha \beta \gamma},\\
    &J_{23}=\dfrac{p\alpha \beta \left(\beta - \alpha \gamma\right)}{\left(-1 + p \alpha \beta \gamma\right)^{2}} + \dfrac{\alpha}{-1 + p \alpha \beta \gamma},\\
    &J_{31}=\dfrac{p\beta \left(v \alpha \beta - \gamma\right)\gamma}{v \left(-1 + p \alpha \beta \gamma\right)^{2}} + \dfrac{\beta}{-1 + p \alpha \beta \gamma},\\
    &J_{32}=-\dfrac{p\alpha \left(v \alpha \beta - \gamma\right)\gamma}{v \left(-1 + p \alpha \beta \gamma\right)^{2}} + \dfrac{\alpha}{-1 + p \alpha \beta \gamma},\\
    &J_{33}=-\dfrac{p\alpha \beta \left(v \alpha \beta - \gamma\right)}{v \left(-1 + p \alpha \beta \gamma\right)^{2}} - \dfrac{1}{v \left(-1 + p \alpha \beta \gamma\right)}.
\end{align}
The $p$-adic absolute value of the determinant of $J$ is
\begin{equation}\label{eq:abdetjacangoli}
    \left|\det J\right|_p= \left| -\frac{(\beta^{2} p - 1)\, (\alpha^{2} v - 1)\, (v - \gamma^{2} p)}{v^{2}\, (\alpha \beta \gamma p - 1)^{4}}\right|_p.
\end{equation}
We also calculate the $p$-adic absolute value of the quadratic form $Q_+^{(4)}(\mathbf{q})$ squared that is, by \eqref{Hmeasure}, the weight (the inverse of the weight) of the Haar measure on $\HH_p^{\times}$
\begin{equation}
    \big\lvert Q_+^{(4)}(\mathbf{q})\big\rvert_p^2 =\left|\frac{(\beta^{2} p + 1)^{2}\, (\alpha^{2} v - 1)^{2}\, (v - \gamma^{2} p)^{2}}{v^{2}\, (\alpha \beta \gamma p - 1)^{4}}\right|_p.
\end{equation}
Finally, after multiplying and dividing by a factor $1=\left|\frac{1}{v}\right|_p$ since $v$ is a $p$-adic unit, produces the following weight:
\begin{equation}\label{SO3weight}
    \frac{\left|\det J\right|_p}{\big\lvert Q_+^{(4)}(\mathbf{q})\big\rvert_p^2}=\left|\frac{1-p\beta^{2}}{1+p\beta^{2}}\frac{1}{(1-v\alpha^{2})(1+p\beta^{2})%(v - p\gamma^{2}) = v*
    (1-\frac{p}{v}\gamma^2)}\right|_p.
\end{equation}
Let us now focus on the quantity $\left\lvert\frac{1-p\beta^{2}}{1+p\beta^{2}}\right\rvert_p$. This is equal to $1$ for $\beta=0$. For every $\beta\in\mathbb{Q}_p^\times$, there exist unique $n\in\mathbb{Z}$ and $u\in\mathbb{U}_p$ such that $\beta=p^nu$, and then $\lvert \pm p\beta^2\rvert_p = p^{-(2n+1)}\neq1$. By the strong triangle inequality, the quantity $\lvert 1\pm p\beta^2\rvert_p=\max\{\lvert1\rvert_p, \lvert \pm p\beta^2\rvert_p\}=\max\{1,p^{-(2n+1)}\}$ is the same for both signs and hence $\lvert 1\pm p\beta^2\rvert_p\neq0$. Lastly, for $\beta=\infty$, we have $\frac{1-p\beta^{2}}{1+p\beta^{2}}=-1$ and then for every $\beta\in\mathbb{Q}_p\cup\{\infty\}$ we have that $\left\lvert\frac{1-p\beta^{2}}{1+p\beta^{2}}\right\rvert_p=1$, which can be simplified in the expression \eqref{SO3weight}.
By using the change of variables theorem, we obtain
\beq 
d\nu(x_1,x_2,x_3)=\frac{d\alpha d\beta d\gamma}{\left\lvert(1-v\alpha^{2})(1+p\beta^{2})(1-\frac{p}{v}\gamma^2)\right\rvert_p}.
\eeq

To find the Haar measure on $\SO(3)_p$, we use the isomorphism $\mathcal{T}_p\colon \HH_p^\times/\QQ_p^\times\to\SO(3)_p$ in Eq.\eqref{totalcorisoQpHpSO3}. Indeed, for any $\phi \in C_c(\SO(3)_p)$, 
\begin{align}
\int_{\SO(3)_p} \phi(\cR) \ d\mu_{\SO(3)_p}(\cR) &= 
\int_{\HH_p^\times/\QQ_p^\times} \phi(\mathcal{T}_p(\upxi\QQ_p^\times)) \ d\mu_{\HH_p^\times/\QQ_p^\times}(\upxi\QQ_p^\times)\nonumber \\
&= 
\int_{U_0} \phi(\mathcal{T}_p(\widetilde{x}\QQ_p^\times)) \  d({\varphi_0}_\ast\nu)(\widetilde{x}\QQ_p^\times)\nonumber \\
&= \int_{D^c} \phi(\mathcal{T}_p(\varphi_0(\ell(\alpha,\beta,\gamma)))) \frac{d\alpha  d\beta  d\gamma}{\left\lvert (1-v\alpha^{2})(1+p\beta^{2})(1-\frac{p}{v}\gamma^2) \right\rvert_p}.
\end{align}
Denoting $\cR(\alpha,\beta,\gamma) = \mathcal{T}_p(\varphi_0(\ell(\alpha,\beta,\gamma)))$, the Haar measure on $\SO(3)_p$ in terms of nautical angles $\alpha,\beta,\gamma$ such that $\mathcal{R}(\alpha,\beta,\gamma)=\mathcal{R}_z(\alpha)\mathcal{R}_y(\beta)\mathcal{R}_x(\gamma)$ is given by 
\beq \label{finalweight}
d\mu_{\SO(3)_p}(\cR(\alpha,\beta,\gamma)) =  \frac{d\alpha d\beta d\gamma}{\left\lvert (1-v\alpha^{2})(1+p\beta^{2})(1-\frac{p}{v}\gamma^2) \right\rvert_p}.
\eeq 
From Eq.\eqref{finalweight} we see that $d\mu_{\SO(3)_p}(\cR(\alpha,\beta,\gamma))$ factorizes as the product of the measures on $\SO(2)_{p,d}$ with $d\in\{-v,p,-p/v\}$ (see e.g. Eq.\eqref{so2measure}). Finally, using Eqs.$\eqref{so2normfact},$ it entails that the normalizing factor for $d\mu_{\SO(3)_p}(\cR(\alpha,\beta,\gamma))$ is
\beq
\left(\mu_{\SO(2)_{p,-v}}(\SO(2)_{p,-v})\cdot \mu_{\SO(2)_{p,p}}(\SO(2)_{p,p})\cdot\mu_{\SO(2)_{p,-p/v}}(\SO(2)_{p,-p/v})\right)^{-1}=\left[4\left(1+\frac{1}{p}\right)\right]^{-1}.
\eeq

\end{proof}

\begin{remark}
The weight \eqref{finalweight} of the Haar measure on $\SO(3)_p$ is very peculiar because it gives us a product measure on three axes and no mixed terms. In particular, since we proved that $\left\lvert\frac{1-p\beta^{2}}{1+p\beta^{2}}\right\rvert_p=1,$ we see that Haar measure \eqref{finalhaarmeasure} on $\SO(3)_p$ factorizes up to the global normalization constant as the product of the Haar measures associated to $\SO(2)_{p,-v},\SO(2)_{p,p},\SO(2)_{p,-\frac{p}{v}}$. 
This factorization is a peculiar feature of the non-Archimedean setting --- in contrast to the Haar measure $d\mu_{\SO(3)_\RR}(R(\alpha,\beta,\gamma)) = \sin\beta d\alpha d\beta d\gamma$ --- arising from the interplay between the rational nature of the nautical parametrization and the features of the $p$-adic absolute value. This could also be traced back to structural results on semisimple Lie groups, as the existence of KAK or Cartan decomposition of $\SO(3)_p$.
\end{remark}

Restricting the Haar measure $\mu_{\SO(3)_p}$ to a subgroup $\SO(2)_{p,d}$ yields the null measure: since $\SO(2)_{p,d}$ has codimension $2$ in $\SO(3)_p$, one has $\mu_{\SO(3)_p}(\SO(2)_{p,d})=0$. Thus, the measure on $\SO(2)_{p,d}$ cannot be recovered by direct geometric restriction of $\mu_{\SO(3)_p}$. Instead, the connection emerges at the level of the local coordinate densities. The $\beta$-dependent factor $\sin\beta d\beta$ in the Haar measure on $\SO(3)_{\RR}$ is different from the Haar measure $d\beta$ on $\SO(2)_\RR$. However, because of the peculiar factorization of the Haar measure on $\SO(3)_p$, its individual factors formally coincide with the expressions of the Haar measures on the respective subgroups $\SO(2)_{p,-v}$, $\SO(2)_{p,p}$, and $\SO(2)_{p,-p/v}$.

\begin{corollary}
Let $\phi\in L^1(\SO(3)_p)$ depending only on $\cR_z(\alpha),\cR_y(\beta),\cR_x(\gamma),$ respectively. Then, the Haar measure relative to rotations around reference axes yields the following Haar integrals:
\beq
\int\limits_{\SO(3)_p}\phi(\cR_z(\alpha)) \ d\mu_{\SO(3)_p}(\cR)=\left(1+\frac{1}{p}\right)^{-1}\int_{\QQ_p\cup \infty} (\phi  \circ  \mathcal{T}_p)(\alpha)  \frac{1}{\left|1-v\alpha^2\right|_p} \ d\alpha,
\eeq
\beq 
\int\limits_{\SO(3)_p}\phi(\cR_y(\beta)) \ d\mu_{\SO(3)_p}(\cR)=\frac{1}{2}\int_{\QQ_p\cup \infty} (\phi  \circ  \mathcal{T}_p)(\beta)  \frac{1}{\left|1+p\beta^2\right|_p} \ d\beta,
\eeq 
\beq
\int\limits_{\SO(3)_p}\phi(\cR_x(\gamma)) \ d\mu_{\SO(3)_p}(\cR)=\frac{1}{2}\int_{\QQ_p\cup \infty} (\phi \circ  \mathcal{T}_p)(\gamma)  \frac{1}{\left|1-\frac{p}{v}\gamma^2\right|_p} \ d\gamma.
\eeq
\end{corollary}

\begin{proof}
    From Theorem~\ref{finalhaarmeasuret} we have that the Haar integral of $\phi$ factorizes as a product of three integrals over $\SO(2)_{p,-v}, \SO(2)_{p,p}$ and $\SO(2)_{p,-p/v}$ with the corresponding measures. From Theorem~5.40 in $\cite{thesis}$ and by using Theorem~\ref{finaliso}, the result follows trivially.
    
\end{proof}

\section{Conclusions}

In this work, we have explicitly constructed the normalized Haar measure on the compact $p$-adic special orthogonal group $\text{SO}(3)_p$ using a nautical (Cardano) parameterization. By exploiting the topological isomorphism between $\text{SO}(3)_p$ and the projective $p$-adic quaternions $\HH_p^\times / \QQ_p^\times$, we derived the corresponding change-of-variables formulas. The exact computation of the Jacobian in the $p$-adic setting demonstrates a highly non-trivial structural property: the Haar measure density factorizes completely across the three nautical angles, without mixed terms. This explicit parameterization furnishes a practical and essential tool for performing integrations over the group manifold, thereby supporting applications in $p$-adic quantum mechanics and quantum information theory, where rotational symmetries are fundamental to the description of spins and qubits. Building upon these results, several compelling open problems emerge. While the quaternionic isomorphism provides an elegant derivation, a natural open question is whether the Haar measure can be computed directly on the $\SO(3)_p$ manifold. This geometric approach requires deriving the $p$-adic invariant differential forms of maximal degree strictly from the nautical matrix parametrization, entirely bypassing the quaternionic algebra to reveal the intrinsic differential geometry of the group. With an explicit angular integral available, it is now possible to perform a rigorous abstract harmonic analysis starting from the regular representation of $\SO(3)_p$. Invoking the Peter-Weyl theorem allows the explicit decomposition of $L^2(\SO(3)_p)$ to compute characters and Clebsch-Gordan coefficients. This mathematical framework provides the analytical backbone to the recent theoretical constructions of $p$-adic qubits, entanglement, and universal quantum logic gates \cite{SIMW26}. Parametrizing higher dimensional $p$-adic special orthogonal groups is highly non-trivial, since standard Euler and nautical angle decompositions are no longer available. Thus, a challenging case of study is adapting the nautical factorization to parameterize rotations in $\SO(4)_p$ and express its related Haar measure. Finally, the structurally anomalous even prime $p=2$ rejects standard Cardano parameterizations due to wild ramification. While one might intuitively suggest to invoke a $p$-adic version of the KAK decomposition \cite{Helgason}, we observe that such a continuous instrument is structurally degenerate here. Indeed, since $\SO(3)_p$ is already a maximal compact subgroup of itself, the abelian factor is trivial, rendering a literal KAK decomposition conceptually inapplicable. Therefore, identifying discrete algebraic factorizations will be crucial for decomposing arbitrary operations in $2$-adic quantum information processing.

%%%%%%%%%%%%%%%%%%%%%%%%%%%%%%%%%%%%%%%%%%%%%%%%%%

\end{document}